\documentclass[a4paper,USenglish,cleveref]{lipics-v2019}

\usepackage{xspace}	
\nolinenumbers
\graphicspath{{./FIGURES/}}

\newcommand{\etal}{{\emph{et~al.}}}
\newcommand{\barA}{{\bar A}}
\newcommand{\barP}{{\bar P}}     
\newcommand{\barX}{{\bar X}}
\newcommand{\bfP}{{\bf{P}}}
\newcommand{\calA}{{\cal A}}
\newcommand{\calP}{{\cal P}}
\newcommand{\calX}{{\cal X}}

\newcommand{\reals}{{\mathbb R}}
\newcommand{\razy}{\,} 
\newcommand{\intermA}{\ddot{A}} 

\newcommand{\AlgDrift}{{\sc{Drift}}\xspace} 
\newcommand{\AlgExtDrift}{{\sc{ExtDrift}}\xspace} 
\newcommand{\ALG}{\ensuremath{\calA}\xspace}
\newcommand{\OPT}{\textsc{Opt}\xspace}

\DeclareMathOperator{\pythagoras}{\vartriangle}
\newcommand{\ltwodist}[2]{\ensuremath{\pythagoras(#1,\,#2)}\,}
\newcommand{\proj}[1]{\ensuremath{#1^x}}
\newcommand{\distsymb}{{\ell}}
\newcommand{\distsub}[2]{{\ensuremath{\distsymb_{#1#2}}}}
\newcommand{\distpar}[2]{{\ensuremath{\distsymb(#1,#2)}}}

\newcommand{\half}{{\textstyle\frac{1}{2}}}
\newcommand{\threehalves}{{\textstyle\frac{3}{2}}}
\newcommand{\sqrttwo}{{\sqrt{2}}}
\newcommand{\sqrtthree}{{\sqrt{3}}}

\title{Better Bounds for Online Line Chasing}
\titlerunning{Better Bounds for Online Line Chasing}

\author{Marcin Bienkowski}{Institute of Computer Science, University of Wrocław, Poland}{marcin.bienkowski@cs.uni.wroc.pl}{https://orcid.org/0000-0002-2453-7772}{}
\author{Jarosław Byrka}{Institute of Computer Science, University of Wrocław, Poland}{jaroslaw.byrka@cs.uni.wroc.pl}{https://orcid.org/0000-0002-3387-0913}{}
\author{Marek Chrobak}{University of California at Riverside, CA, USA}{marek@cs.ucr.edu}{https://orcid.org/0000-0002-8673-2709}{}
\author{Christian Coester}{University of Oxford, United Kingdom}{christian.coester@cs.ox.ac.uk}{https://orcid.org/0000-0003-3744-0977}{}
\author{Łukasz Jeż}{Institute of Computer Science, University of Wrocław, Poland}{lukasz.jez@cs.uni.wroc.pl}{https://orcid.org/0000-0002-7375-0641}{} 
\author{Elias Koutsoupias}{University of Oxford, United Kingdom}{elias@cs.ox.ac.uk}{https://orcid.org/0000-0002-2226-6737}{}

\authorrunning{M.~Bienkowski, J.~Byrka, M.~Chrobak, C.~Coester, Ł.~Jeż, E.~Koutsoupias}
\funding{Research supported by NSF grant CCF-1536026, by Polish National Science Centre grants 2016/22/E/ST6/00499 and 2015/18/E/ST6/0045, and by ERC Advanced Grant 321171 (ALGAME).}

\Copyright{Marcin Bienkowski, Jarosław Byrka, Marek Chrobak, Christian Coester, Łukasz Jeż, Elias Koutsoupias}

\ccsdesc{Theory of computation~Online algorithms}
\keywords{convex body chasing, line chasing, competitive analysis}

\EventEditors{Peter Rossmanith, Pinar Heggernes, and Joost-Pieter Katoen}
\EventNoEds{3}
\EventLongTitle{44th International Symposium on Mathematical Foundations of Computer Science (MFCS 2019)}
\EventShortTitle{MFCS 2019}
\EventAcronym{MFCS}
\EventYear{2019}
\EventDate{August 26--30, 2019}
\EventLocation{Aachen, Germany}
\EventLogo{}
\SeriesVolume{138}
\ArticleNo{58}

\begin{document}

\maketitle   

\begin{abstract}
We study online competitive algorithms for the \emph{line chasing problem} in
Euclidean spaces $\reals^d$, where the input consists of an initial point
$P_0$ and a sequence of lines $X_1,X_2,...,X_m$, revealed one at a~time. At
each step $t$, when the line $X_t$ is revealed, the algorithm must determine a
point $P_t\in X_t$. An~online algorithm is called $c$-competitive if for any
input sequence the path $P_0, P_1,...,P_m$ it computes has length at most $c$
times the optimum path. The line chasing problem is a variant of a~more
general convex body chasing problem, where the sets $X_t$ are arbitrary convex
sets.

To date, the best competitive ratio for the line chasing problem was $28.1$,
even in the plane. We improve this bound by providing
a simple~$3$-competitive algorithm for any dimension~$d$. We complement this bound by a matching lower bound for algorithms that are memoryless in the sense of our algorithm, and a lower
bound of $1.5358$ for arbitrary algorithms. The latter bound also improves upon the previous lower bound of $\sqrt{2}\approx 1.412$ for convex body chasing in $2$ dimensions.
\end{abstract}


\section{Introduction}
\label{sec: introduction}

\emph{Convex body chasing} is a fundamental problem in online
computation. It asks for an incrementally-computed path that traverses a given
sequence of convex sets provided one at a time in an online fashion and is as short as possible.
Formally, the input consists of an initial point $P_0 \in \reals^d$ and
a sequence $X_1, X_2, ..., X_m \subseteq \reals^d$ of convex sets.
The objective is to find a~path $\bfP = (P_0,P_1,...,P_m)$ with $P_t\in X_t$ for each
$t=1,2,...,m$ and minimum total length $\distsymb(\bfP) = \sum_{t=1}^m \distsub{P_{t-1}}{P_t}$.
(Throughout the paper, by $\distpar{P}{Q}$ or 
$\distsub{P}{Q}$ we denote the Euclidean distance between points $P$ and $Q$ in $\reals^d$.) 
This path $\bfP$ must be computed \emph{online}, in the following sense: the sets
$X_t$ are revealed over time, one per time step. At step~$t$, when set $X_t$ is
revealed, we need to immediately and irrevocably identify its visit point $P_t \in X_t$.
Thus the choice of $P_t$ does not depend on the future sets $X_{t+1},...,X_m$.

As can be easily seen, in this online scenario computing an optimal solution is not possible,
and thus all we can hope for is to find a path whose length only approximates the optimum value.
A widely accepted measure for the quality of this approximation is the competitive ratio.
For a constant $c\ge 1$,
we will say that an online algorithm $\calA$ is \emph{$c$-competitive} if it computes
a path whose length is at most $c$ times the optimum solution (computed offline).
This constant $c$ is called the \emph{competitive ratio} of $\calA$.
Our objective is then to design an online algorithm whose competitive ratio is as close to $1$ as possible.

The convex body chasing problem was originally introduced in 1993 by
Friedman and Linial~\cite{Friedman_convex_body_1993}, who gave a constant-competitive
algorithm for chasing convex bodies in $\reals^2$ (the plane) and conjectured that it is possible
to achieve constant competitiveness in any $d$-dimensional space $\reals^d$. 
As shown in~\cite{Friedman_convex_body_1993},
this constant would have to depend on $d$; in fact it needs to be at least $\sqrt{d}$.   

The Friedman-Linial conjecture has remained open for over two decades. In the last several years 
this topic has experienced a sudden increase in research activity, 
partly motivated by connections to machine learning (see \cite{Bansal_nested_covex_2018,Bubeck_competitively_chasing_2018}),
resulting in rapid progress.
In 2016, Antoniadis~{\etal}~\cite{Antoniadis_chasing_convex_2016}
gave a $2^{O(d)}$-competitive algorithm for chasing affine spaces of any dimension.
In 2018, Bansal~{\etal}~\cite{Bansal_nested_covex_2018}
gave an~algorithm with competitive ratio $2^{O(d \log d)}$ for nested families of convex sets, where the
input set sequence satisfies $X_1\supseteq X_2 \supseteq ... \supseteq X_m$. Soon later
their bound was improved to $O(d\log d)$ by Argue~{\etal}~\cite{Argue_nearly-linear_2018},
and then to $O(\sqrt{d\log d})$ by Bubeck~{\etal}~\cite{Bubeck_chasing_nested_2018}.
Finally, Bubeck~{\etal}~\cite{Bubeck_competitively_chasing_2018} just recently announced a proof of the Friedman-Linial conjecture, providing an algorithm with competitive ratio~$2^{O(d)}$ for arbitrary convex sets.

One other natural variant of convex body chasing that also attracted attention in the literature is
\emph{line chasing}, where all sets $X_t$ are lines.  
Friedman and Linial~\cite{Friedman_convex_body_1993} gave an online algorithm for line chasing
in $\reals^2$ with ratio $28.53$.
Their algorithm was simplified by Antoniadis~{\etal}~\cite{Antoniadis_chasing_convex_2016},
who also slightly improved the ratio, to $28.1$. Earlier, in 2014,
Sitters~\cite{Sitters_generalized_work_function_2014} showed that a generalized
work function algorithm has constant competitive ratio for line chasing, but he did not 
determine the value of the constant.


\subsection{Our results}

We study the line chasing problem discussed above. We give
a $3$-competitive algorithm for line chasing in $\reals^d$, for any dimension $d\ge 2$,
significantly improving the competitive ratios from~\cite{Friedman_convex_body_1993,Antoniadis_chasing_convex_2016,Sitters_generalized_work_function_2014}.
Our algorithm is very simple and essentially memoryless, as it only needs to keep track of the
last line in the request sequence. We start by providing the algorithm for line chasing in the plane, in 
\cref{sec: 3-competitive algorithm in the plane},
and later in \cref{sec: an algorithm for arbitrary dimension} we extend it to an arbitrary dimension. In Section~\ref{sec:memoryless}, we provide a matching lower bound of $3$ for algorithms that are memoryless in the sense stated above and oblivious with respect to rotation, translation and uniform scaling of the metric space.
We also provide a lower bound for arbitrary algorithms (see \cref{sec: lower bound}), showing that no online algorithm
can achieve competitive ratio better than $1.5358$. 
This improves the lower bound of $\sqrt{2}\approx 1.412$ for line chasing established in~\cite{Friedman_convex_body_1993}, which was previously also the best known lower bound for the more general problem of convex body chasing in the plane.


\subsection{Other related work} 

Set chasing problems are also known as \emph{Metrical Service Systems} (see below) and belong to a very general class of problems for online optimization and competitive analysis called 
\emph{Metrical Task Systems (MTS)}~\cite{Borodin_optimal_online_1992}.
An instance of MTS specifies a metric space $M$, an initial point $P_0 \in M$, and 
a~sequence of non-negative functions $\tau_1,\tau_2,...,\tau_m$ over $M$ called \emph{tasks}.
These tasks arrive online, one at a time. At each step $t$,
the algorithm needs to choose a point $P_t \in M$ where it moves to ``process'' the current task $\tau_t$.
The goal is to minimize the total cost defined by
$\sum_{t=1}^m ( \mu(P_{t-1},P_t) + \tau_t(P_t) )$, where $\mu()$  is the metric in $M$.
Thus in MTS, in                                        
addition to movement cost, at each step we also pay the cost of ``processing''~$\tau_t$. 
For any metric space $M$ with $n$ points,
if we allow arbitrary non-negative task functions then
a competitive ratio of $2n-1$ can be achieved and is optimal. 
This general bound is not particularly useful, because in many online optimization problems that can be modeled 
as an MTS, the metric space $M$ has additional structure and only tasks of some special form are allowed, which
makes it possible to design online algorithms with constant competitive ratios, independent of the size of $M$.

An MTS where $M = \reals^d$ and all functions $\tau_t$ are convex
is referred to as \emph{convex function chasing}, and was studied 
in~\cite{Antoniadis_chasing_convex_2016,Bansal_online_convex_2015,Lin_dynamic_right-sizing_2011}.
For the special case of convex functions on the real line, a $2$-competitive algorithm was given in \cite{Bansal_online_convex_2015}.

An MTS where each task function $\tau_t$ takes value $0$ on a subset $X_t\subseteq M$ and $\infty$ elsewhere is called a \emph{Metrical Service System (MSS)}~\cite{Chrobak_metrical_task_systems_1996}. In other words, in an MSS, in each step $t$ the algorithm needs to move to a point in $X_t$. To achieve a competitive ratio independent of the size of $M$, it is generally required to restrict the sets $X_t$ to be in some subset $\calX\subsetneq\calP(M)$. For instance, finite competitive ratios can be achieved when $\calX$ is the set of sets of size at most $k$ \cite{FiatFKRRV98,Burley96,Ramesh95}. If $M=\reals^d$ and $\calX$ is the set of convex subsets, this is precisely the convex body chasing problem, and if $\calX$ is the set of lines, it is the line chasing problem. One variant of MSS that has been particularly well studied is
the famous \emph{$k$-server problem} (see, for example, \cite{Manasse_competitive_algorithms_1990,Koutsoupias_k-server_1995}),
in which one needs to schedule movement of                          
$k$ servers in response to requests arriving online in a metric space, where
each request must be covered by one server. 
(In the MSS representation of the $k$-server problem, 
each set $X_t$ consists of all $k$-tuples of points that include the request point at step~$t$.)


\section{A 3-Competitive Algorithm in the Plane}
\label{sec: 3-competitive algorithm in the plane}

\begin{figure}[t]
\begin{center} 
\includegraphics[width = 2.5in]{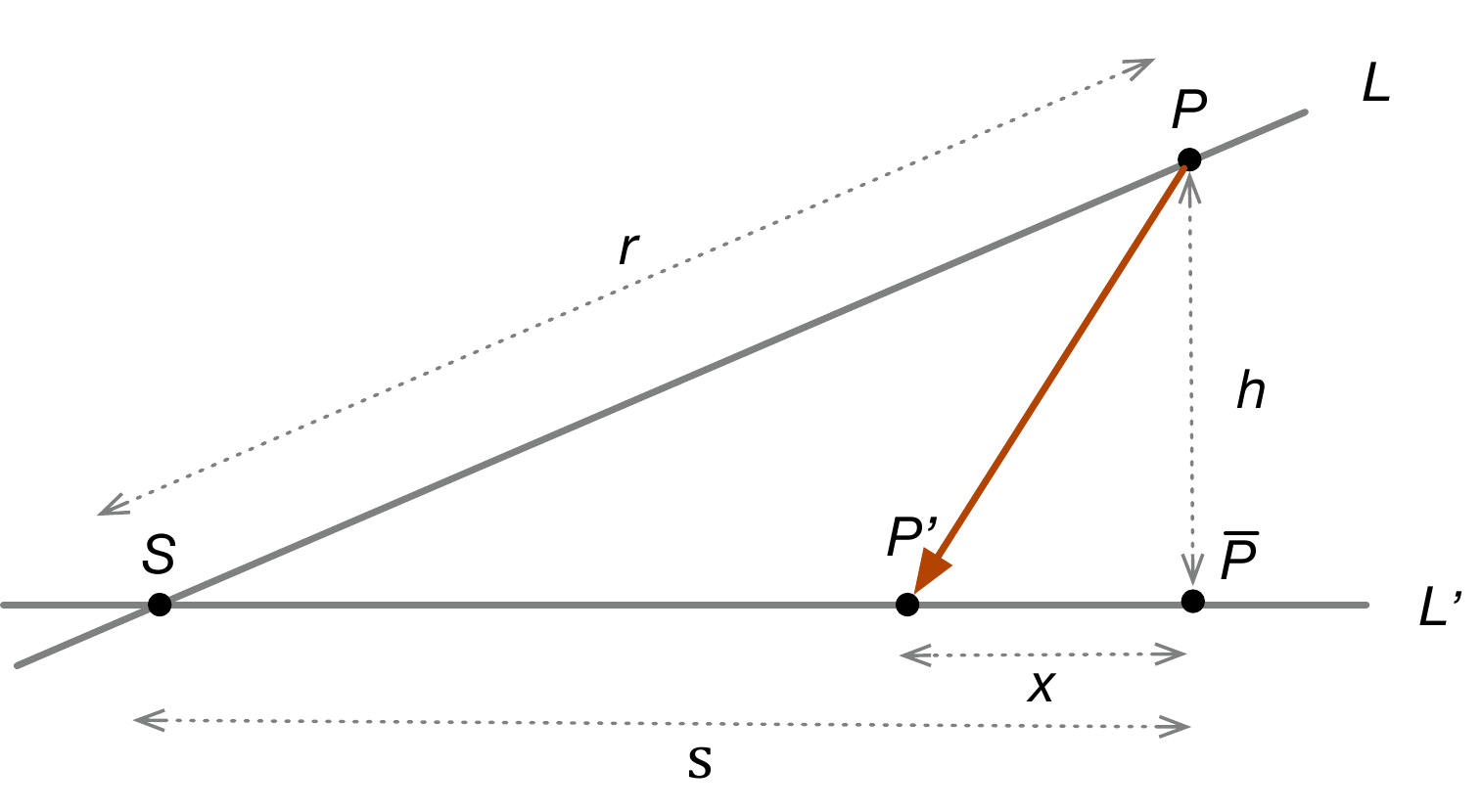}
\caption{Algorithm~{\AlgDrift} moves from $P$ to $P'$.}
\label{fig: algorithm drift}
\end{center}
\end{figure}
	
In this section, we present our online algorithm for line chasing in $\reals^2$ with
competitive ratio~$3$. The intuition is this: suppose that the last requested line is~$L$ and that
the algorithm moved to point $P\in L$. Let $L'$ be the new request line, $S$ the intersection
point of $L$ and $L'$, and $r = \distsub{S}{P}$. A na{\"i}ve greedy algorithm would move to the point $\barP$ on $L'$ 
nearest to $P$ (see \cref{fig: algorithm drift}) at cost $h = \distsub{P}{\barP}$.
If $h$ is small, then $r- \distsub{S}{\barP} = o(h)$, that is the distance between the
greedy algorithm's point and $S$ decreases only by a negligible amount.
But the adversary can move to $S$, paying cost $r$, and then
alternate requests on $L$ and $L'$. On this sequence the overall cost of this algorithm
would be $\omega(r)$, so it would not be constant-competitive. This example shows that if the angle between $L$ and $L'$
is small then the drift distance towards $S$ needs to be roughly proportional to $h$. 
Our algorithm is designed so that this distance is roughly $h/\sqrt{2}$ if $h$ is small (with the coefficient
chosen to optimize the competitive ratio), and that it becomes $0$ when $L'$ is perpendicular to $L$.


\begin{center}
\begin{minipage}{0.95\textwidth}
\hrule
\smallskip
\textbf{Algorithm~\AlgDrift}
Suppose that the last request is line $L$ and that the algorithm is on point $P\in L$.
Let the new request be $L'$ and for any point $X\in L$, let $\barX$ be the orthogonal
projection of $X$ onto $L'$.  If $L'$ does not intersect $L$, move to $P'=\barP$.
Otherwise, let $S = L\cap L'$ be the intersection point of $L$ and $L'$.
Let also $r = \distsub{S}{P}$, $h = \distsub{P}{\barP}$, and $s = \distsub{S}{\barP}$
(see \cref{fig: algorithm drift}).
Move to point $P'\in L'$ such that $\distsub{S}{P'} = s - x$,
where $x \;=\; \frac{1}{\sqrttwo} (h+s-r)$. 
\smallskip  
\hrule       
\end{minipage}
\end{center}


\begin{theorem}
Algorithm \AlgDrift is $3$-competitive for the line chasing problem in~$\reals^2$.
\end{theorem}

\begin{proof}
We establish an upper bound on the competitive ratio via amortized analysis, based on
a potential function. The (always non-negative) value of this potential function, $\Phi(P,A)$, depends on
locations $P, A \in L$ of the algorithm's and the adversary's point on the current
line $L$. If $L'$ is the new request line, and $P',A'\in L'$ are
the new locations of the algorithm's and adversary's points, we want this function to satisfy
\begin{eqnarray}
	\distsub{P}{P'} + \Phi(P',A') - \Phi(P,A) \;\le\;  3 \razy \distsub{A}{A'}. 
	\label{eqn: potential function property}
\end{eqnarray}
Since initially the potential is $0$ and is always non-negative, adding inequality~\eqref{eqn: potential function property}
for all moves will establish $3$-competitiveness of Algorithm~{\AlgDrift}.  

The potential function we use in our proof is $\Phi(P,A) = \sqrtthree \razy \distsub{A}{P}$.
Substituting this formula, inequality~\eqref{eqn: potential function property} reduces to 
\begin{equation}
	\distsub{P}{P'} + \sqrtthree \razy (\, \distsub{A'}{P'} - \distsub{A}{P} \, ) \;\le\; 3 \razy \distsub{A}{A'}.
	\label{eqn: amortized bound}
\end{equation}
It thus remains to prove inequality~\eqref{eqn: amortized bound}.
Let $g = \distsub{A}{\barA}$, $z = \distsub{A'}{\barA}$, and $v = \distsub{\barA}{\barP}$.

We first discuss the trivial case of non-intersecting $L$ and $L'$.
Keeping with the general notation, here we have $x=0$ and thus $\distsub{P}{P'}=h$.
Moreover, $g=\distsub{A}{\barA}=h$ as well.  For fixed $z$, we have $\distsub{A}{A'} = \sqrt{h^2+z^2}$,
i.e., the right hand side of \eqref{eqn: amortized bound} is fixed,
whereas the left hand side is maximized if $A'$ is on the other side of $\barA$ than~$\barP$.
The left hand side is thus at most
\begin{equation*}
	h + \sqrt{3} \razy z \leq \sqrt{2} \razy \sqrt{h^2+3z^2} 
		\;\leq\; \sqrt{2} \razy \sqrt{3\razy (h^2+z^2)} = \sqrt{6} \razy \distsub{A}{A'} < 3 \razy \distsub{A}{A'},
\end{equation*}
where the first inequality follows from the power mean inequality (for powers $1$ and $2$),
proving this easy case.

\begin{figure}[t]
\centering
\includegraphics[width = 3in]{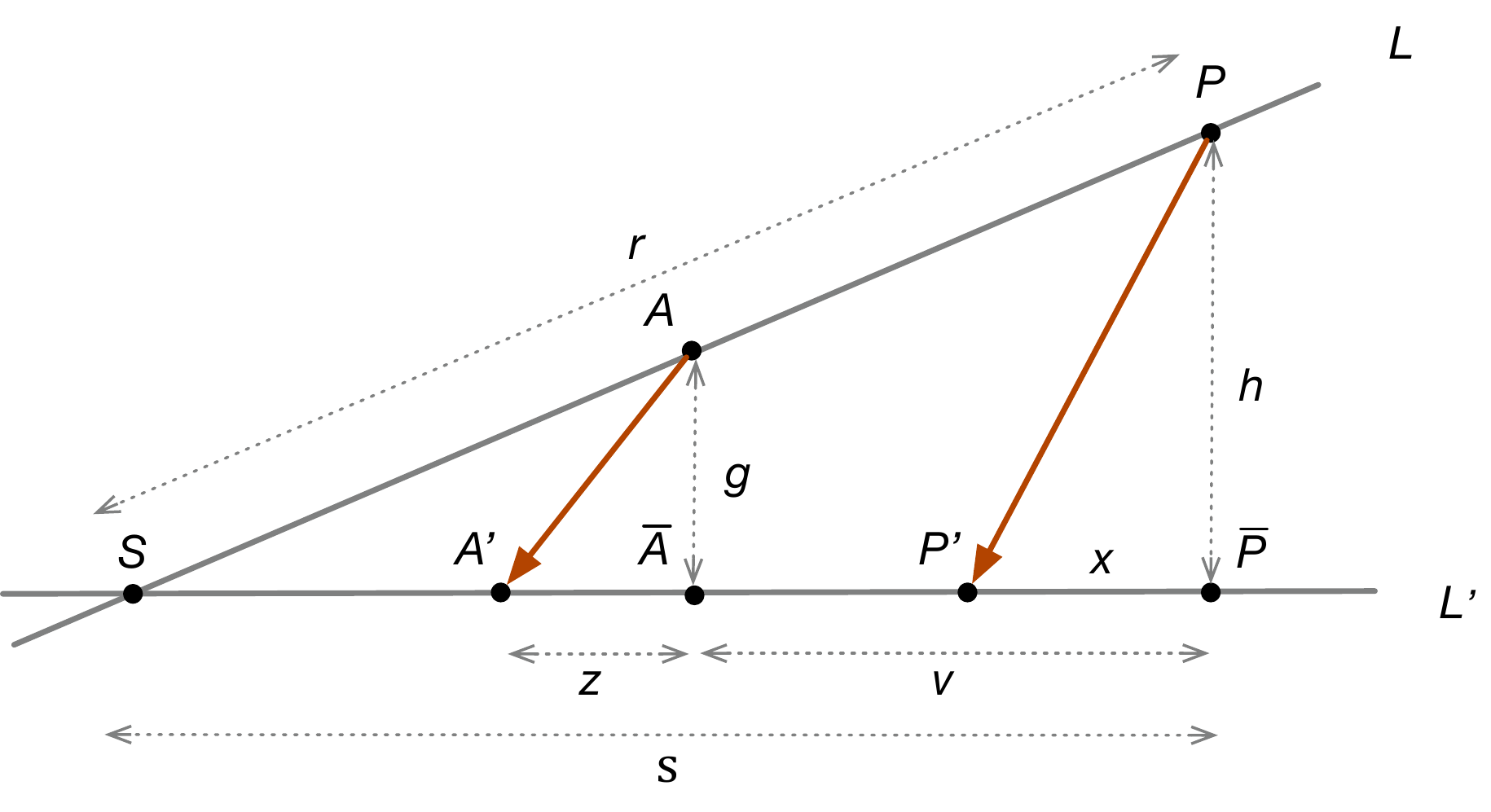}
\caption{Notation for the analysis of Algorithm~{\AlgDrift}.}
\label{fig: upper_bound_3_notation}
\end{figure}

The situation when $L'$ and $L$ do intersect is illustrated in \cref{fig: upper_bound_3_notation}. 
(The figure shows only the case when $\barA$ is between $S$ and $P'$.)
Orient $L'$ from left to right (with $\barP$ being to the right of $S$), as shown in this figure. 
We want to express the distances in the above inequality in terms
of $s$, $h$, $v$, and $z$ (keeping in mind that $x$ and $r$ are functions of $h$ and $s$):
\begin{align*} 
\distsub{P}{P'} \;&=\; \sqrt{x^2+h^2}
\\
\distsub{A}{P}   \;&=\; vr/s \;=\; (v \sqrt{s^2+h^2})/s
\\
\distsub{A}{A'}   \;&=\; \sqrt{z^2+g^2} 
\end{align*}
The values of $g$ and $\distsub{A'}{P'}$ depend on some cases, that we consider below.


\begin{description}

\item[Case 1.]
$\barA$ is between $S$ and $P'$,
as in \cref{fig: upper_bound_3_notation}. 
Then $g = h(s-v)/s$.
Our goal is first to find $A'$ for which the bound in~\eqref{eqn: amortized bound} is tightest.
For a given $z$,  
among the two locations of $A'$ at distance $z$ from $\barA$, the one on the
left gives a larger value of the left-hand side of~\eqref{eqn: amortized bound}, while 
the right-hand side is the same for both.
Thus we can assume that $A'$ is to the left of $\barA$, so $\distsub{A'}{P'} = z+v-x$.
Then we can rewrite~\eqref{eqn: amortized bound} as follows:
\begin{equation}   
	\textstyle
   \frac{1}{\sqrtthree} \razy  \distsub{P}{P'} - \distsub{A}{P} + v - x
	 		\;\le \sqrtthree \razy \sqrt{z^2+g^2} -  z 
	\label{eqn: amortized bound case 1}
\end{equation}
By elementary calculus, the right-hand side is minimized for $z = \frac{1}{\sqrttwo}\razy g$,
so we can assume that $z$ has this value.  Then inequality~\eqref{eqn: amortized bound case 1} reduces
to
\begin{equation}   
	\textstyle
   \frac{1}{\sqrtthree} \razy  \distsub{P}{P'} - \distsub{A}{P} + v - x
	 		\;\le  \sqrttwo\razy g.
	\label{eqn: amortized bound case 1 reduced}
\end{equation}
After substituting $g = h(s-v)/s$ and  $\distsub{A}{P} = vr/s$, 
inequality~\eqref{eqn: amortized bound case 1 reduced} reduces
further to
\begin{equation}   
	\textstyle
   s\razy ( \frac{1}{\sqrtthree} \razy \distsub{P}{P'} -x - \sqrttwo h )
	 		\;\le v \razy ( r - s - \sqrttwo\razy h).
	\label{eqn: amortized bound case 1 reduced more}
\end{equation}
The expression in the parenthesis on    
the right-hand side of~\eqref{eqn: amortized bound case 1 reduced more} is non-positive
by triangle inequality,
so the right-hand side is minimized when $v$ is maximized, that is $v = s$,
and then it reduces to 
\begin{equation}   
	\textstyle
   x^2+h^2  \;\le 3 ( r - s + x )^2.
	\label{eqn: amortized bound case 1 reduced even more}
\end{equation}
Recall that $x =  \frac{1}{\sqrttwo} (h+ s- r )$. Since $r-h\le s \le r$, we have
\begin{align*}   
	\textstyle
   x^2+h^2   \;&=\;   \half (h+ s- r )^2 + h^2
			\\
			\;&\le\;  \half h^2 + h^2
 			\\
	 \textstyle   	\;&=\; \threehalves h^2
 					\;\le\; \threehalves [\, h + (\sqrt{2}-1) (r- s) \,]^2
 						\;=\; 3( r - s + x )^2,
\end{align*}
proving~\eqref{eqn: amortized bound case 1 reduced even more}.


\item[Case 2.]
$\barA$ is before $S$. In this case we have $g = h(v-s)/s$.
Just as in Case~1, we can assume that  $A'$ is to the
left of $\barA$, so that $\distsub{A'}{P'} = z+v-x$, 
and~\eqref{eqn: amortized bound} reduces to
\begin{equation}   
	\textstyle
   \frac{1}{\sqrtthree} \razy  \distsub{P}{P'} - \distsub{A}{P} + v - x
	 		\;\le  \sqrttwo\razy g.
	\label{eqn: amortized bound case 2 reduced}
\end{equation}
After substituting $g = h(v-s)/s$ and  $\distsub{A}{P} = vr/s$, 
inequality~\eqref{eqn: amortized bound case 1 reduced} reduces
further to
\begin{equation}   
	\textstyle
   s\razy ( \frac{1}{\sqrtthree} \razy \distsub{P}{P'} -x + \sqrttwo h )
	 		\;\le v \razy ( r - s + \sqrttwo\razy h).
	\label{eqn: amortized bound case 2 reduced more}
\end{equation}
The expression in the parenthesis on the right-hand side of \eqref{eqn: amortized bound case 2 reduced more}
is non-negative, so the right-hand side is minimized when $v = s$ (because
in this case $v\ge s$), so \eqref{eqn: amortized bound case 2 reduced more} reduces to
the same inequality~\eqref{eqn: amortized bound case 1 reduced even more} as in
Case~1, completing the argument for Case~2.


\item[Case 3.]
$\barA$ is after $\barP$. In this case we have $g = h(v+s)/s$. 
Symmetrically to Case~1, we can now assume that $A'$ is to the right of $\barA$, so that
$\distsub{A'}{P'} = z+v+x$, and that $z = \frac{1}{\sqrttwo}\razy g$.
Then, analogously to~\eqref{eqn: amortized bound case 1 reduced},
 we can rewrite~\eqref{eqn: amortized bound} as follows:
\begin{equation}   
	\textstyle
   \frac{1}{\sqrtthree} \razy  \distsub{P}{P'} - \distsub{A}{P} + v + x
	 		\;\le  \sqrttwo\razy g
	\label{eqn: amortized bound case 3 reduced}
\end{equation}
After substituting $g = h(v+s)/s$ and  $\distsub{A}{P} = vr/s$, 
inequality~\eqref{eqn: amortized bound case 3 reduced} reduces
further to
\begin{equation}   
	\textstyle
   s\razy ( \frac{1}{\sqrtthree} \razy \distsub{P}{P'} + x - \sqrttwo h )
	 		\;\le v \razy ( r - s + \sqrttwo\razy h).
	\label{eqn: amortized bound case 3 reduced more}
\end{equation}
The expression in the parenthesis on    
the right-hand side of~\eqref{eqn: amortized bound case 3 reduced more} is non-negative,
so the right-hand side is minimized when $v=0$, and then it reduces to 
\begin{equation}   
	\textstyle
   x^2+h^2  \;\le 3 ( \sqrttwo h - x )^2.
	\label{eqn: amortized bound case 3 reduced even more}
\end{equation}
To prove this, we proceed similarly as in Case~1:
\begin{equation*}   
	\textstyle
   x^2+h^2   \;\le\; \threehalves h^2
 	\;\le\; \threehalves ( h + r- s)^2
 					\;=\; 3 ( \sqrttwo h - x )^2,
\end{equation*}
proving~\eqref{eqn: amortized bound case 3 reduced even more}.


\item[Case 4.]
$\barA$ is between $P'$ and $\barP$. Then $g = h(s-v)/s$ (as in Case~1).
Similar to Case~3, we can assume that $A'$ is to the right of $\barA$, so that now
$\distsub{A'}{P'} = z-v+x$, and that $z = \frac{1}{\sqrttwo}\razy g$.   
Then, analogously to~\eqref{eqn: amortized bound case 1 reduced},
 we can rewrite~\eqref{eqn: amortized bound} for this case as follows:
\begin{equation}   
	\textstyle
   \frac{1}{\sqrtthree} \razy  \distsub{P}{P'} - \distsub{A}{P} - v + x
	 		\;\le  \sqrttwo\razy g
	\label{eqn: amortized bound case 4 reduced}
\end{equation}
After substituting $g = h(s-v)/s$ and  $\distsub{A}{P} = vr/s$, 
inequality~\eqref{eqn: amortized bound case 4 reduced} reduces
further to               
\begin{equation}   
	\textstyle
   s\razy ( \frac{1}{\sqrtthree} \razy \distsub{P}{P'} + x - \sqrttwo h )
	 		\;\le v  \razy ( r + s - \sqrttwo\razy h).
	\label{eqn: amortized bound case 4 reduced more}
\end{equation}
We now have two sub-cases. If the expression in the parenthesis on    
the right-hand side of~\eqref{eqn: amortized bound case 4 reduced more} is non-negative
then 
the right-hand side is minimized when $v=0$, so inequality~\eqref{eqn: amortized bound case 4 reduced more}
reduces to inequality~\eqref{eqn: amortized bound case 3 reduced even more} from Case~3.
If this expression is negative (that is when $r+s <\sqrt{2} h$), 
then it is sufficient to prove \eqref{eqn: amortized bound case 4 reduced more} with
$v$ on the right-hand side replaced by $s$ (because $v\le s$). This reduces it to
$\frac{1}{\sqrtthree} \razy \distsub{P}{P'} + x \le r + s$. This last inequality follows
from $\distsub{P}{P'}\le r$ and $x\le s$.          
\qedhere
\end{description}
\end{proof}


\section{An Algorithm for Arbitrary Dimension}
\label{sec: an algorithm for arbitrary dimension}

In this section, we show how to extend Algorithm~{\AlgDrift}
to Euclidean spaces $\reals^d$ for arbitrary
dimension $d\ge 2$. This extension, that we call~{\AlgExtDrift},  
is quite simple, and consists of projecting the whole space onto
an appropriately chosen plane that contains the new request line.
While such approach was suggested already by Friedman and Linial~\cite{Friedman_convex_body_1993},
their choice of plane may lose a constant factor in the competitive ratio.
We project onto a different plane, which allows {\AlgExtDrift} to also be $3$-competitive.

Let $P$ be the current {\AlgExtDrift} position and $L'$ the new request line.
If $P \in L'$, {\AlgExtDrift} makes no move.  Otherwise, let $U$ be the uniquely determined plane
which contains both $L'$ and $P$.  {\AlgExtDrift} makes the move prescribed by {\AlgDrift} in the plane $U$
for $P$, $L'$ and the projection of $L$ onto $U$.

\begin{theorem}
Algorithm {\AlgExtDrift} is $3$-competitive for the line chasing problem in $\reals^d$, for arbitrary dimension 
$d \geq 2$.
\end{theorem}

\begin{proof}
We prove that \eqref{eqn: potential function property} holds in arbitrary dimension.
If $P\in L'$ then $L$ and $L'$ are co-planar, so the analysis from the previous section
works directly. 

So assume that $P\notin L'$.
We first allow the adversary to perform a free move from its current
position $A$ to point $\intermA$ defined as the orthogonal projection of $A$ onto $U$, and then
we analyze the move within $U$ (that is, in a two-dimensional setting), 
as if the adversary started from point $\intermA$.

We note that $\distsub{\intermA}{X} \leq \distsub{A}{X}$ for any point $X \in U$, as
$(\distsub{A}{X})^2 = (\distsub{\intermA}{X})^2 + (\distsub{A}{\intermA})^2$ by definition of $\intermA$.
It follows that:
\begin{itemize}
	\item In the free adversary move from $A$ to $\intermA$ the potential function decreases
	(by taking $X=P$ in the above inequality)  and both costs are $0$.    
	Further, in the move within $U$, with the adversary starting from $\intermA$, 
	Algorithm~{\AlgExtDrift} makes the same move as {\AlgDrift}, which implies 
	that \eqref{eqn: potential function property} is satisfied.
	Thus the complete move (combining the free adversary move and the move
	inside $U$) satisfies inequality~\eqref{eqn: potential function property} as well.  
	
	\item The free move is only beneficial for the adversary: taking $X=A'$ shows that
	the cost of moving to $A'$ from $\intermA$ is no more costly for the adversary
	than moving to $A'$ from $A$.\qedhere
\end{itemize}
\end{proof}


\section{Lower Bound for Memoryless Algorithms}
\label{sec:memoryless}

We show that our algorithm achieves the optimal competitive ratio among a certain class of ``memoryless'' algorithms. For a metric space $M$, let $\calX\subseteq\calP(M)$ be the set of possible requests (i.e., lines in our case). In general, we can view an algorithm as a function $\calA\colon M\times\calX^*\to\calX$ with $\calA(P_0)=P_0$ and $\calA(P_0,X_1,\dots,X_t)\in X_t$ for each initial point $P_0\in M$ and requests $X_1,\dots,X_t\in\calX$. We call an algorithm \emph{memoryless} if $\calA(P_0,X_1,\dots,X_t)$ is a function of only the last position $\calA(P_0,X_1,\dots,X_{t-1})$, the last request $X_{t-1}$ and the new request $X_t$.

However, memorylessness alone would not impose any limit on the power of line-chasing algorithms: By perturbing its positions very slightly, an algorithm could always encode the entire history in low significant bits of its current position. To get a meaningful notion of memorylessness, we therefore require an additional property, namely that the algorithm is oblivious with respect to rotation, translation or scaling of the metric space. More precisely, a \emph{direct similarity} of $\reals^d$ is a bijection $f\colon\reals^d\to\reals^d$ that is a composition of rotation, translation and scaling by some factor $r_f>0$. In particular, for any $P,Q\in\reals^d$, we have $\distpar{f(P)}{f(Q)}=r_f\distpar{P}{Q}$.
We call an algorithm $\calA$ \emph{rts-oblivious} if $\calA(f(P_0),f(X_1),\dots,f(X_t))=f(\calA(P_0,X_1,\dots,X_t))$ for any $P_0\in M$, $X_i\in\calX$ and any direct similarity $f$. In general (when algorithms are allowed to use memory) there is no reason to behave differently when the input is transformed by such $f$, since it is just a renaming of points and scaling of distances by a uniform constant. For completeness, we provide a proof of this intuition via the following proposition:

\begin{proposition}\label{prop:rts-oblivious}
If there is a $c$-competitive algorithm for line-chasing, then there is a $c$-competitive rts-oblivious algorithm.
\end{proposition}
\begin{proof}
	For an initial position $P_0$ and request sequence $X_1,\dots,X_t$, we assume without loss of generality that $P_0\notin X_1$. For any such $P_0$ and $X_1$, there exists a unique direct similarity $g=g_{P_0X_1}$ such that $g(P_0)=(0,1)$ and $g(X_1)=\reals\times\{0\}$. Given a $c$-competitive algorithm $\calA$, we claim that the algorithm $\tilde{\calA}$ given by \begin{align*}
	\tilde{\calA}(P_0,X_1,\dots,X_t)=g^{-1}(\calA(g(P_0),g(X_1),\dots,g(X_t)))
	\end{align*}
	is rts-oblivious and $c$-competitive.

To see that $\tilde{\calA}$ is rts-oblivious, consider an arbitrary direct similarity $f$. Notice that $g_{f(P_0)f(X_1)}=g\circ f^{-1}$. Thus,
\begin{align*}
\tilde{\calA}(f(P_0),f(X_1),\dots,f(X_t)) &= (f\circ g^{-1})(\calA(g(P_0),g(X_1),\dots,g(X_t)))\\
&= f(\tilde{\calA}(P_0,X_1,\dots,X_t)),
\end{align*}
as required. To see that $\tilde{\calA}$ is $c$-competitive, consider an initial position $P_0$ and request sequence $X_1,\dots,X_m$ along with an adversary's solution $A_0=P_0,A_1\in X_1,\dots,A_m\in X_m$. The cost of $\tilde{\calA}$ can be bounded via
\begin{align*}
\sum_{t=1}^m&\distpar{\tilde{\calA}(P_0,X_1,\dots,X_{t-1})}{\tilde{\calA}(P_0,X_1,\dots,X_{t})}\\
&= \frac{1}{r_{g}}\sum_{t=1}^m\distpar{{\calA}(g(P_0),g(X_1),\dots,g(X_{t-1}))}{{\calA}(g(P_0),g(X_1),\dots,g(X_{t}))}\\
&\le \frac{c}{r_{g}}\sum_{t=1}^m\distpar{g(A_{t-1})}{g(A_t)}\\
&= c\sum_{t=1}^m\distpar{A_{t-1}}{A_t},
\end{align*}
where the inequality uses that $\calA$ is $c$-competitive against the solution $g(A_0),\dots,g(A_m)$ for the transformed input $g(P_0),g(X_1),\dots,g(X_m)$.
\end{proof}

Intuitively, an rts-oblivious algorithm does not know the absolute coordinates of its positions and requests, but only relative to each other and up to scaling. If it is memoryless, in the plane this boils down to only knowing the angle between the new and the old request line. We show now that our algorithms \AlgDrift and \AlgExtDrift achieve the optimal competitive ratio among rts-oblivious memoryless algorithms.

\begin{theorem}
	Any rts-oblivious memoryless algorithm for line-chasing has competitive ratio at least $3$.
\end{theorem}
\begin{proof}
We will construct an initial point $P_0$ and lines $L_0,\dots,L_m$ in $\reals^2$ with the property that $P_0\in L_0$ and $L_t$ can be obtained by rotating $L_{t-1}$ around some point $S_t\in L_{t-1}$ in clockwise direction by less than $90$ degrees.

Let $P_0,\dots,P_m$ be the sequence of points visited by a given algorithm. We use notation similar to that in \cref{fig: algorithm drift}: Write $\barP_{t-1}$ for the orthogonal projection of $P_{t-1}$ onto $L_t$ and let $h_t=\distpar{P_{t-1}}{\barP_{t-1}}$ and $s_t=\distpar{\barP_{t-1}}{S_t}$. The movement from $P_{t-1}$ to $P_t$ can always be viewed as first moving to $\barP_{t-1}$ and then moving some distance $x_t\in\reals$ in the direction towards intersection $S_t$, for a total cost $\sqrt{h_t^2+x_t^2}$. Here, $x_t<0$ would constitute movement away from $S_t$ and $x_t>s_t$ would constitute movement beyond $S_t$.

Observe that for rts-oblivious memoryless algorithms, $\frac{x_t}{h_t}$ is a function of only $\frac{h_t}{s_t}$, i.e. $\beta(\frac{h_t}{s_t})=\frac{x_t}{h_t}$ for some function $\beta\colon(0,\infty)\to\reals$. Any rts-oblivious memoryless algorithm for line-chasing in the plane is uniquely determined by its associated function $\beta$ as well as similar functions for the cases of counter-clockwise rotations of at most 90 degrees and parallel lines.\footnote{If we require algorithms to be oblivious also with respect to reflection (which would still satisfy \cref{prop:rts-oblivious}), they would be uniquely determined by $\beta$ alone. \AlgDrift is the algorithm corresponding to $\beta(a)=\frac{a+1-\sqrt{a^2+1}}{\sqrt{2}a}$.} Let $\beta(0):=\limsup_{a\to0}\beta(a)\in\reals\cup\{-\infty,\infty\}$. Let us first show that algorithms with $\beta(0)=\infty$ or $\beta(0)\le 0$ have unbounded competitive ratio.

If $\beta(0)=\infty$, we choose $P_0=(1,h)$, $L_0=\{(x,y)\colon y=hx\}$, $L_1=\reals\times\{0\}$ for some small $h>0$. The algorithm's cost is $h\sqrt{1+\beta(h)^2}$, whereas the optimal cost is $h$. Choosing $h$ arbitrarily small shows that the competitive ratio is unbounded.

If $\beta(0)\le0$, fix some $\epsilon\in(0,\frac{1}{2}]$ and choose $a\in(0,\epsilon]$ with $\beta(a)\le\epsilon$. Let $P_0=(1,0)$, $L_0=\reals\times\{0\}$ and define $L_t$ as the clockwise rotation of $L_{t-1}$ around the origin $O:=(0,0)$ by angle $\arctan(a)$. Thus, we have $\frac{h_t}{s_t}=a$ for each $t$. Notice that $s_t\sqrt{1+a^2}=\distpar{O}{P_{t-1}}=s_{t-1}(1-a\beta(a))$, and therefore
\begin{align*}
\frac{s_t}{s_{t-1}} = \frac{1-a\beta(a)}{\sqrt{1+a^2}}\ge 1-\frac{a^2+a\beta(a)}{1+a^2} \ge 1-2\epsilon a,
\end{align*}
where the first inequality uses $\sqrt{1+a^2}\le 1+a^2$ and the second inequality uses $0<a\le\epsilon$ and $\beta(a)\le \epsilon$. Hence,
\begin{align*}
s_t\ge (1-2\epsilon a)^{t-1}
\end{align*}
Since $\distpar{P_{t-1}}{P_t}\ge h_t=as_t$, the total cost of the algorithm is
\begin{align*}
\sum_{t=1}^m\distpar{P_{t-1}}{P_t} \ge a\sum_{t=0}^{m-1} (1-2\epsilon a)^t \xrightarrow{m\to\infty}\frac{1}{2\epsilon}.
\end{align*}
Meanwhile, an optimal algorithm pays total cost $1$ by moving to $O$ immediately. Letting $\epsilon\to0$, we find again that the competitive ratio is unbounded.

It remains to consider the case $0<\beta(0)<\infty$. Then we can choose arbitrarily small $a>0$ such that $0<a\beta(a)<1$. We choose the initial point $P_0=(\frac{1}{a},1)$, and the request sequence starts with $L_0=\{(x,y)\colon y=ax\}$ and $L_1=\reals\times \{0\}$. For $t\ge 2$, we define $L_t$ as the clockwise rotation of $L_{t-1}$ around $S_t=S_2=\left(\frac{1}{a}-\beta(a)+\sqrt{1+a^2}\left(\beta(a)+\frac{1}{2\beta(a)}\right),0\right)$ by angle $\arctan(a)$. The idea is that in response to $L_1$, the algorithm drifts to the left (towards intersection $S_1=(0,0)$), but the subsequent requests are such that it would have been cheaper to drift to the right (away from $S_1$) instead.

We have $s_1=\frac{1}{a}$ and $s_2=\frac{\distpar{P_1}{S_2}}{\sqrt{1+a^2}}=\beta(a)+\frac{1}{2\beta(a)}$. For $t\ge 3$, similarly to the previous case we get
\begin{align*}
\frac{s_t}{s_{t-1}}\ge 1-\frac{a^2+a\beta(a)}{1+a^2} \ge 1-a^2-a\beta(a)
\end{align*}
and therefore
\begin{align*}
s_t\ge\left(\beta(a)+\frac{1}{2\beta(a)}\right)\left(1-a^2-a\beta(a)\right)^{t-2}\qquad\text{if }t\ge 2.
\end{align*}
As $m\to\infty$, the cost of the algorithm is
\begin{align*}
\sum_{t=1}^\infty\distpar{P_{t-1}}{P_t} &= \sum_{t=1}^\infty h_t\sqrt{1+\beta(a)^2} = \sqrt{1+\beta(a)^2}a\sum_{t=1}^\infty s_t\\
&\ge \sqrt{1+\beta(a)^2}\left(1+\left(\beta(a)+\frac{1}{2\beta(a)}\right)\frac{1}{a+\beta(a)}\right)\\
&\xrightarrow{a\to0}\sqrt{1+\beta(0)^2}\left(2+\frac{1}{2\beta(0)^2}\right),
\end{align*}
where the limit $a\to0$ is taken along a sequence where $\beta(a)\to \beta(0)$.
In contrast, an offline algorithm can move immediately from $P_0$ to $S_2$, paying cost $\sqrt{1+\frac{1}{4\beta(0)^2}}$ as $a\to 0$ and $\beta(a)\to\beta(0)$. By dividing, we see that the competitive ratio is at least
\begin{align*}
\sqrt{(1+\beta(0)^2)\left(4+\frac{1}{\beta(0)^2}\right)}=\sqrt{4\beta(0)^2+\frac{1}{\beta(0)^2}+5},
\end{align*}
which is minimized for $\beta(0)=\frac{1}{\sqrt{2}}$, taking value $3$.
\end{proof}


\section{Lower Bound for Arbitrary Algorithms}
\label{sec: lower bound}

Finally, in this section, we show how to improve an existing lower bound 
of $\sqrt{2} \approx 1.41$ for arbitrary algorithms to $1.5358$. Our bound holds even in two dimensions, and improves also the lower bound for the more general convex body chasing in two dimensions.

\begin{theorem}
The competitive ratio of any deterministic online algorithm \ALG for the line 
chasing problem is at least $1.5358$. 
\end{theorem}

\begin{proof}
We describe our adversarial strategy below. On the created input, we will compare 
the cost of \ALG to the cost of an offline optimum $\OPT$. We
assume that both \ALG and \OPT start at origin point $P_0 = A_0 = (0,0)$.

Our construction is parameterized with real positive numbers 
$c_1 = 0.5535$, 
$c_2 = 0.4965$, 
$c_3 = 0.8743$, 
$a_1 = 1.3012$, 
$a_2 = 0.6663$, 
$p_2 = 0.5612$, and 
$p_3 = 0.1696$.

We fix points $P_1 = (0,c_1)$, $C_2 = (0,c_1+c_2)$, $C_3 = (0,c_1+c_2+c_3)$ and $A_3 = (1,c_1)$,
see \cref{fig:improved_lb} for illustration.
For succinctness, we use notation $\ltwodist{x}{y} = \sqrt{x^2 + y^2}$.

\begin{figure}[t]
\centering
\includegraphics[width = 2.8in]{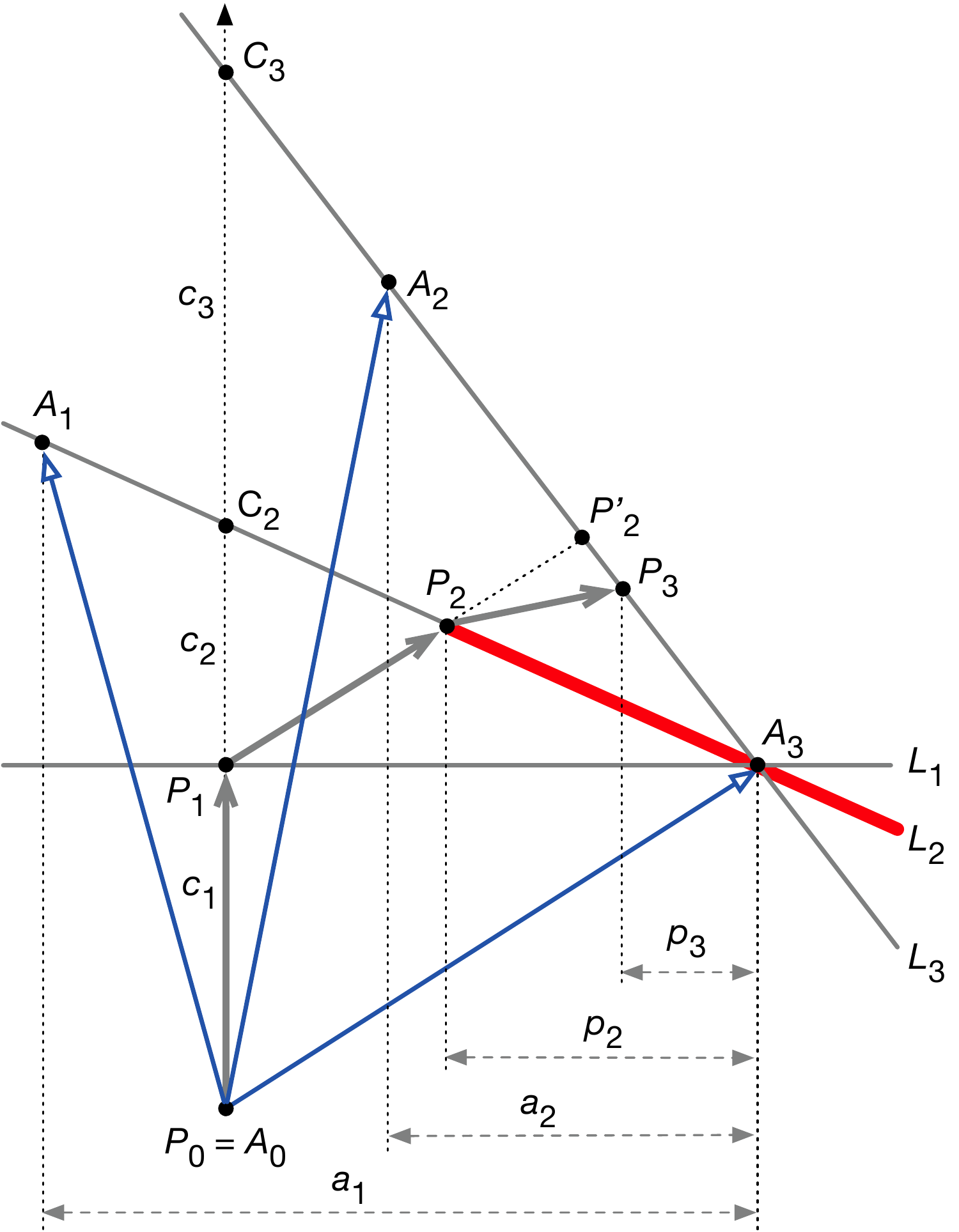}
\caption{Visual description of our lower bound for arbitrary algorithms. Lines
$L_1$, $L_2$ and $L_3$ are presented to an~online algorithm. Blue arrows
describe possible movements of \OPT, while gray thick arrows describe a~path
of an~algorithm that minimizes the competitive ratio for this adversarial
construction. Red thick half-line denotes the forbidden region.}
\label{fig:improved_lb}
\end{figure}

\paragraph*{Initial part: Line $L_1$}

The first request line is the line $P_1 A_3$, denoted $L_1$. Without loss of
generality, we can assume that \ALG moves to point $P_1$. This is because
the adversary can either play the strategy described below or its mirror
image (flipped against the line $P_0P_1$), so any deviation from $P_1$, either to
the left or right, can only increase the cost of \ALG. 

From now on, for any point $Q$ we denote its projection on line $L_1$ by
$\proj{Q}$.

\paragraph*{Middle part: Line $L_2$} 

Next, the adversary issues the request line $C_2 A_3$, denoted~$L_2$. 
Let $P_2 \in L_2$ and $A_1 \in L_2$ be the points to the left of $A_3$, such that 
$\distsub{\proj{P_2}}{A_3} = p_2$ and $\distsub{\proj{A_1}}{A_3} = a_1$.

Let $\bar{P_2}$ be the point on $L_2$ chosen by \ALG. If $\bar{P_2}$ lies to the right 
of point $P_2$, then the adversary forces \ALG to move to $A_1$ (by giving
sufficiently many different lines that go through $A_1$ at different angles).
\OPT may then serve the whole sequence by going from $A_0$ to $A_1$ at cost 
\begin{align*}
	\distsub{A_0}{A_1}
		= &\; \ltwodist{c_1+c_2 \cdot a_1}{a_1-1} \leq 1.23679
\intertext{while the cost of \ALG is then at least}
	\distsub{P_0}{P_1} + \distsub{P_1}{P_2} + \distsub{P_2}{A_1} 
		= &\; \distsub{P_0}{P_1} + \distsub{P_1}{P_2} + \distsub{A_3}{A_1} - \distsub{A_3}{P_2} \\
		= &\; c_1 + \ltwodist{1-p_2}{c_2\cdot p_2} 
		+ \ltwodist{a_1}{c_2 \cdot a_1} - \ltwodist{p_2}{c_2 \cdot p_2} \\
		\geq &\; 1.89948
\end{align*}
Hence, the competitive ratio in this case is at least $1.5358$. 

We call the half-line of $L_2$ to the right of point $P_2$ \emph{forbidden
region}. From now on, we assume that the point chosen by \ALG in $L_2$
does not lie in this region.

\paragraph*{Final part: Line $L_3$} 

Finally, the adversary issues the request line $C_3 A_3$, denoted~$L_3$. 
Let $P'_2$ be the intersection of line $P_1 P_2$ with line $L_3$. 
Next, let $A_2$ and $P_3$ be the points on the line $L_3$ to the left of $A_3$, such that
$\distsub{\proj{A_2}}{A_3} = a_2$ and $\distsub{\proj{P_3}}{A_3} = p_3$. Note that $P_3$ belongs to the
interval $P'_2A_3$.

Let $\bar{P_3}$ be the point on $L_3$ chosen by \ALG. 
We consider two cases. 

\begin{description}

\item[Case 1.]
$\bar{P_3}$ lies at point $P_3$ or to its left. In this case, the adversary forces \ALG to
move to~$A_3$. \OPT may serve the whole sequence by going from $A_0$ to $A_3$
paying
\[
	\distsub{A_0}{A_3} = \ltwodist{1}{c_1} \leq 1.142963.
\]
We may now argue that the cost of \ALG is minimized if $\bar{P_3}$ is equal to $P_3$:
If $\bar{P_3}$ is to the left of point $P'_2$, then the cost 
of \ALG is at least $\distsub{P_0}{P_1} + \distsub{P_1}{\bar{P_3}} + \distsub{\bar{P_3}}{A_3}$. Both the second and the third
summand decrease when we move $\bar{P_3}$ towards $P'_2$. Hence, now we may assume that 
$\bar{P_3}$ belongs to the interval $P'_2P_3$. 
As the path of \ALG must avoid forbidden region, its cost is at least 
$\distsub{P_0}{P_1} + \distsub{P_1}{P_2} + \distsub{P_2}{\bar{P_3}} + \distsub{\bar{P_3}}{A_3}$. 
The sum of the last two summands decreases when we move $\bar{P_3}$ towards $P_3$. 
Therefore, we obtain that the cost of \ALG is at least 
\begin{align*}
	\distsub{P_0}{P_1} + & \distsub{P_1}{P_2} + \distsub{P_2}{P_3} + \distsub{P_3}{A_3} \\
		= &\; c_1 + \ltwodist{1-p_2}{c_2 \cdot p_2} 
			+ \ltwodist{(c_2+c_3) \cdot p_3 - c_2 \cdot p_2}{p_2 - p_3} \\
			& \quad\quad + \ltwodist{p_3}{(c_2+c_3) \cdot p_3} 
		\geq 1.75537.
\end{align*}
Thus, in this case the competitive ratio is at least $1.5358$.

\item[Case 2.]
If $\bar{P_3}$ lies to the right of point $P_3$, then the adversary forces \ALG to
move to $A_2$. \OPT may serve the whole sequence by going from $A_0$ to $A_2$
at cost
\[
	\distsub{A_0}{A_2} = \ltwodist{c_1+(c_2+c_3) \cdot a_2}{1-a_2} \leq 1.50435.
\]
To go from $P_1$ to $\bar{P_3}$ and avoid the forbidden region, \ALG
has to pay at least $\distsub{P_1}{P_2} + \distsub{P_2}{\bar{P_3}}$. 
Therefore, its cost is at least 
\begin{align*}
	\distsub{P_0}{P_1} + & \distsub{P_1}{P_2} + \distsub{P_2}{\bar{P_3}} + \distsub{\bar{P_3}}{A_2}  \\
		\geq &\; \distsub{P_0}{P_1} + \distsub{P_1}{P_2} + \distsub{P_2}{P_3} + \distsub{P_3}{A_2} \\
		\geq &\; \distsub{P_0}{P_1} + \distsub{P_1}{P_2} + \distsub{P_2}{P_3} + \distsub{A_2}{A_3} - \distsub{P_3}{A_3} \\
		= &\; c_1 + \ltwodist{1-p_2}{c_2 \cdot p_2} 
			+ \ltwodist{(c_2+c_3) \cdot p_3 - c_2 \cdot p_2}{p_2 - p_3} \\
			& \quad \quad + \ltwodist{a_2}{(c_2+c_3) \cdot a_2} 
			- \ltwodist{p_3}{(c_2+c_3) \cdot p_3} 
			\geq 2.31039.
\end{align*}
Thus, in this case the ratio is also at least $1.5358$.
\qedhere
\end{description}
\end{proof}


\section{Final Comments}
\label{sec: final comments}

Establishing the optimal competitive ratio for line chasing with memory remains an open
problem. We believe that with memory, a competitive ratio better than $3$ is achievable.

The intuition is that in the first move,
if $L$ and $P$ are the initial
line and position and $L'$ is the new request line, then the 
algorithm should move to the nearest point $\barP$ on $L'$. More generally,
if the requests on $L$ and $L'$ alternate (and their angle is small), the
algorithm should initially drift slowly towards $S = L\cap L'$ and 
only gradually accelerate as it becomes more credible that the adversary is located at $S$. To gauge this credibility for general request sequences, an algorithm might store the current work function at each step.

It appears also that our lower bound of $1.5358$ can be improved  by introducing additional steps, 
although this gives only very small improvements and leads to a very involved analysis. 
It is possible that an approach fundamentally different from ours may give a
better bound with simpler analysis.


\bibliographystyle{plainurl}
\bibliography{convex_body_chasing}

\begin{thebibliography}{10}

\bibitem{Antoniadis_chasing_convex_2016}
Antonios Antoniadis, Neal Barcelo, Michael Nugent, Kirk Pruhs, Kevin Schewior,
  and Michele Scquizzato.
\newblock Chasing convex bodies and functions.
\newblock In {\em Proc. 12th Latin American Theoretical Informatics Symposium
  (LATIN)}, pages 68--81, 2016.
\newblock \href {http://dx.doi.org/10.1007/978-3-662-49529-2_6}
  {\path{doi:10.1007/978-3-662-49529-2_6}}.

\bibitem{Argue_nearly-linear_2018}
C.~J. Argue, S{\'{e}}bastien Bubeck, Michael~B. Cohen, Anupam Gupta, and
  Yin~Tat Lee.
\newblock A~nearly-linear bound for chasing nested convex bodies.
\newblock In {\em Proc. 30th ACM-SIAM Symp. on Discrete Algorithms (SODA)},
  pages 117--122, 2019.

\bibitem{Bansal_nested_covex_2018}
Nikhil Bansal, Martin B{\"{o}}hm, Marek Eli{\'{a}}s, Grigorios Koumoutsos, and
  Seeun~William Umboh.
\newblock Nested convex bodies are chaseable.
\newblock In {\em Proc. 29th ACM-SIAM Symp. on Discrete Algorithms (SODA)},
  pages 1253--1260, 2018.
\newblock \href {http://dx.doi.org/10.1137/1.9781611975031.81}
  {\path{doi:10.1137/1.9781611975031.81}}.

\bibitem{Bansal_online_convex_2015}
Nikhil Bansal, Anupam Gupta, Ravishankar Krishnaswamy, Kirk Pruhs, Kevin
  Schewior, and Clifford Stein.
\newblock A 2-competitive algorithm for online convex optimization with
  switching costs.
\newblock In {\em Proc. Approximation, Randomization, and Combinatorial
  Optimization. Algorithms and Techniques (APPROX/RANDOM)}, pages 96--109,
  2015.
\newblock \href {http://dx.doi.org/10.4230/LIPIcs.APPROX-RANDOM.2015.96}
  {\path{doi:10.4230/LIPIcs.APPROX-RANDOM.2015.96}}.

\bibitem{Borodin_optimal_online_1992}
Allan Borodin, Nathan Linial, and Michael~E. Saks.
\newblock An optimal on-line algorithm for metrical task system.
\newblock {\em J. {ACM}}, 39(4):745--763, 1992.
\newblock \href {http://dx.doi.org/10.1145/146585.146588}
  {\path{doi:10.1145/146585.146588}}.

\bibitem{Bubeck_chasing_nested_2018}
S{\'{e}}bastien Bubeck, Yin~Tat Lee, Yuanzhi Li, and Mark Sellke.
\newblock Chasing nested convex bodies nearly optimally.
\newblock {\em CoRR}, abs/1811.00999, 2018.
\newblock URL: \url{http://arxiv.org/abs/1811.00999}.

\bibitem{Bubeck_competitively_chasing_2018}
S{\'{e}}bastien Bubeck, Yin~Tat Lee, Yuanzhi Li, and Mark Sellke.
\newblock Competitively chasing convex bodies.
\newblock In {\em Proc. 51st ACM Symp. on Theory of Computing (STOC)}, pages
  861--868, 2019.
\newblock \href {http://dx.doi.org/10.1145/3313276.3316314}
  {\path{doi:10.1145/3313276.3316314}}.

\bibitem{Burley96}
William~R. Burley.
\newblock Traversing layered graphs using the work function algorithm.
\newblock {\em J. Algorithms}, 20(3):479--511, 1996.
\newblock \href {http://dx.doi.org/10.1006/jagm.1996.0024}
  {\path{doi:10.1006/jagm.1996.0024}}.

\bibitem{Chrobak_metrical_task_systems_1996}
Marek Chrobak and Lawrence~L. Larmore.
\newblock Metrical task systems, the server problem and the work function
  algorithm.
\newblock In {\em Online Algorithms, The State of the Art (Proc. Dagstuhl
  Seminar, June 1996)}, pages 74--96, 1996.
\newblock \href {http://dx.doi.org/10.1007/BFb0029565}
  {\path{doi:10.1007/BFb0029565}}.

\bibitem{FiatFKRRV98}
Amos Fiat, Dean~P. Foster, Howard~J. Karloff, Yuval Rabani, Yiftach Ravid, and
  Sundar Vishwanathan.
\newblock Competitive algorithms for layered graph traversal.
\newblock {\em {SIAM} J. Comput.}, 28(2):447--462, 1998.
\newblock \href {http://dx.doi.org/10.1137/S0097539795279943}
  {\path{doi:10.1137/S0097539795279943}}.

\bibitem{Friedman_convex_body_1993}
Joel Friedman and Nathan Linial.
\newblock On convex body chasing.
\newblock {\em Discrete {\&} Computational Geometry}, 9:293--321, 1993.
\newblock \href {http://dx.doi.org/10.1007/BF02189324}
  {\path{doi:10.1007/BF02189324}}.

\bibitem{Koutsoupias_k-server_1995}
Elias Koutsoupias and Christos~H. Papadimitriou.
\newblock On the k-server conjecture.
\newblock {\em J. {ACM}}, 42(5):971--983, 1995.
\newblock \href {http://dx.doi.org/10.1145/210118.210128}
  {\path{doi:10.1145/210118.210128}}.

\bibitem{Lin_dynamic_right-sizing_2011}
Minghong Lin, Adam Wierman, Lachlan L.~H. Andrew, and Eno Thereska.
\newblock Dynamic right-sizing for power-proportional data centers.
\newblock {\em {IEEE/ACM} Trans. Netw.}, 21(5):1378--1391, 2013.
\newblock \href {http://dx.doi.org/10.1109/TNET.2012.2226216}
  {\path{doi:10.1109/TNET.2012.2226216}}.

\bibitem{Manasse_competitive_algorithms_1990}
Mark~S. Manasse, Lyle~A. McGeoch, and Daniel~Dominic Sleator.
\newblock Competitive algorithms for server problems.
\newblock {\em J. Algorithms}, 11(2):208--230, 1990.
\newblock \href {http://dx.doi.org/10.1016/0196-6774(90)90003-W}
  {\path{doi:10.1016/0196-6774(90)90003-W}}.

\bibitem{Ramesh95}
H.~Ramesh.
\newblock On traversing layered graphs on-line.
\newblock {\em J. Algorithms}, 18(3):480--512, 1995.
\newblock \href {http://dx.doi.org/10.1006/jagm.1995.1019}
  {\path{doi:10.1006/jagm.1995.1019}}.

\bibitem{Sitters_generalized_work_function_2014}
Ren{\'{e}} Sitters.
\newblock The generalized work function algorithm is competitive for the
  generalized 2-server problem.
\newblock {\em {SIAM} J. Comput.}, 43(1):96--125, 2014.
\newblock \href {http://dx.doi.org/10.1137/120885309}
  {\path{doi:10.1137/120885309}}.

\end{thebibliography}


\end{document}